\documentclass{article}
\usepackage{setspace}
\usepackage{braket}
\usepackage{bbm}
\usepackage{tcolorbox}
\usepackage{graphicx}
\usepackage{amsthm}
\usepackage{multicol}
\usepackage{hyperref}
\hypersetup{
    colorlinks,
    citecolor= orange,
    filecolor= blue,
    linkcolor= blue,
    urlcolor= green
}
\usepackage{mathrsfs}
\usepackage{appendix}
\usepackage{amssymb}
\usepackage{relsize}
\usepackage{amsthm}
\usepackage{ amssymb }
\usepackage{amsmath}
\usepackage[utf8]{inputenc}
\usepackage[english]{babel}
\usepackage[a4paper, total={6in, 9.2in}]{geometry}
\usepackage{amsmath}

\newtheorem{definition}{Definition}
\newtheorem{definition2}{Lemma}

\newtheorem{definition4}{Theorem}
\newtheorem{definition5}{Remark}
\newtheorem{definition6}{Claim}
\newtheorem{definition7}{Conjecture}
\newtheorem{definition9}{Proposition}
\newtheorem{Co}{Corollary}
\usepackage[T1]{fontenc}
\usepackage{lmodern}

\usepackage{tikz-cd}

\begin{titlepage}

\title{SBSCV}
\author{Alberto Acevedo}

\date{2023}

\end{titlepage}

\begin{document}
\thispagestyle{plain}
\begin{center}
   
    \textbf{Asymptotic Quantum State Discrimination for Mixtures of Unitarily Related States.}

    \vspace{0.4cm}
    
    \textbf{Alberto Acevedo$^{1}$, Janek Wehr $^{1,2}$}

      \vspace{0.4cm}
      \textbf{$^{1}$Program in Applied Mathematics, The University of Arizona, Tucson, USA, 85721}

      \vspace{0.4cm}
      \textbf{$^{2}$Department of Mathematics, The University of Arizona, Tucson, USA, 85721}

    \vspace{0.9cm}
    \textbf{Abstract:}
    Given a mixture of states, finding a way to optimally \emph{discriminate} its elements is a prominent problem in quantum communication theory. In this paper, we will address mixtures of density operators that are unitarily equivalent via elements of a one-parameter unitary group, and the corresponding \emph{quantum state discrimination} (QSD) problems. We will be particularly interested in QSD as time goes to infinity. We first present an approach to QSD in the case of countable mixtures and address the respective asymptotic QSD optimization problems, proving necessary and sufficient conditions for minimal error to be obtained in the asymptotic regime (we say that in such a case QSD is \emph{fully solvable}). We then outline an analogous approach to uncountable mixtures, presenting some conjectures that mirror the results presented for the cases of countable mixtures. As a technical tool, we prove and use an infinite dimensional version of the well-known Barnum-Knill bound \cite{knill}. 

\end{center}
\section{Introduction}
\;\;\; Quantum State Discrimination (QSD) is the problem of minimizing the error in distinguishing between the elements of a mixture of density operators $\sum_{i}p_{i}\boldsymbol{\hat{\rho}}_{i}$. To define what is meant by \emph{distinguishing} we first introduce the concept of a POVM---\emph{positive operator-valued measure} \cite{Nielsen}, and  \emph{quantum measurement}; we note that all of the Hilbert spaces in this paper will be assumed to be separable. 
\begin{definition}{\textbf{POVM}:}
Consider a Hilbert space $\mathscr{H}$. A trace-preserving POVM is a set of positive semi-definite operators $\big\{\mathbf{\hat{M}}_{i}^{\dagger}\mathbf{\hat{M}}_{i}\big\}_{i}$ acting in $\mathscr{H}$ that sum to the identity operator. i.e. 
\begin{equation}
\label{eqn:res}
\sum_{i}\mathbf{\hat{M}}_{i}^{\dagger}\mathbf{\hat{M}}_{i} = \mathbb{I}_{\mathscr{H}}
\end{equation}
\end{definition}
\begin{definition5}
The POVM may consist of an uncountable set of semi-definite operators as well. In this case the analogous set of operators, e.g. $\mathbf{\hat{M}}_{x}^{\dagger}\mathbf {\hat{M}}_{x}$ ($x\in\mathbb{R}$) must satisfy the analogous constraint as (\ref{eqn:res}), i.e.
\begin{equation}
\label{eqn:pvmqm3}
\int\mathbf{\hat{M}}_{x}^{\dagger}\mathbf{\hat{M}}_{x}dx = \mathbb{I}_{\mathscr{H}}
\end{equation}
\end{definition5}
\begin{definition}{\textbf{Quantum Measurement}:}
Let $\mathcal{S}\big(\mathscr{H}\big)$ be the set of density operators acting in $\mathscr{H}$. Let $\boldsymbol{\hat{\rho}}\in \mathcal{S}\big(\mathscr{H}\big)$ be a system's state. The state of this system after a measurement is
\begin{equation}
\sum_{i}\mathbf{\hat{M}}_{i}\boldsymbol{\hat{\rho}}\mathbf{\hat{M}}_{i}^{\dagger}
\end{equation}

\end{definition}

The QSD optimization problem \cite{hellstrom} \cite{Montanaro} \cite{qiu} \cite{bae} \cite{barnett} may now be defined. Let $\mathcal{S}(\mathcal{\mathscr{H}}\big)$ be the space of density operators acting in $\mathscr{H}$. Given a mixture of density operators, 
\begin{equation}
\label{eqn:countablemixture}
\boldsymbol{\hat{\rho}} = \sum_{i=1}^{N}p_{i} \boldsymbol{\hat{\rho}}_{i} 
\end{equation}
where $\sum_{i = 1}^{N}p_{i} = 1$, the theory of QSD aims to find a POVM $\{\mathbf{\hat{M}}_{l}^{\dagger}\mathbf{\hat{M}}_{l}\}_{l=1}^{K}\subset \mathcal{B}(\mathscr{H})$ ( $K\geq N$) which resolves the identity operator of $\mathcal{B}(\mathscr{H})$, and minimizes the quantity below, referred to as the \emph{probability error}. 
\begin{equation}
\label{eqn:minerror2}
p_{E}\big\{\{p_{i},\boldsymbol{\hat{\rho}}_{i}\}_{i=1}^{N}, \{\mathbf{\hat{M}}_{l}\big\}_{l=1}^{K} \big\}:=1-\sum_{i=1}^{N}p_{i}Tr\big\{\mathbf{\hat{M}}_{i}\boldsymbol{\hat{\rho}}_{i}\mathbf{\hat{M}}^{\dagger}_{i} \big\} 
\end{equation}

To see what the significance of (\ref{eqn:minerror2}) is, let us consider the unread measurement state (\ref{eqn:pvmqm3}) corresponding to the mixture $\sum_{i=1}^{N}\boldsymbol{\hat{\rho}}_{i}$ after a measurement generated by the POVM $\big\{\mathbf{\hat{M}}^{\dagger}_{i}\mathbf{\hat{M}}_{i}\big\}_{i=1}^{N}$.
\begin{equation}
\label{eqn:cupii2}
\sum_{j=1}^{N}\mathbf{\hat{M}}_{j}\Big(\sum_{i=1}^{N}p_{i} \boldsymbol{\hat{\rho}}_{i}\Big)\mathbf{\hat{M}}_{j}^{\dagger} = 
\end{equation}
\begin{equation}
\label{eqn:cupii}
\sum_{i=1}^{N}p_{i}\mathbf{\hat{M}}_{i}\boldsymbol{\hat{\rho}}_{i}\mathbf{\hat{M}}_{i}^{\dagger} + \sum_{j=1}^{N}\sum_{i; j\neq i}^{N}p_{i} \mathbf{\hat{M}}_{j}\boldsymbol{\hat{\rho}}_{i}\mathbf{\hat{M}}_{j}^{\dagger} 
\end{equation}
From the definition of POVM , it is clear that $Tr\big\{\sum_{j=1}^{N}\mathbf{\hat{M}}_{j}\Big(\sum_{i=1}^{N}p_{i} \boldsymbol{\hat{\rho}}_{i}\Big)\mathbf{\hat{M}}_{j}^{\dagger}\big\} = 1$, hence from (\ref{eqn:cupii2}) and (\ref{eqn:cupii})
\begin{equation}
\label{eqn:cuppis3}
1 = \sum_{i=1}^{N}p_{i}Tr\big\{\mathbf{\hat{M}}_{i}\boldsymbol{\hat{\rho}}_{i}\mathbf{\hat{M}}_{j}^{\dagger}\big\} +  \sum_{j=1}^{N}\sum_{i; i\neq j}^{N}p_{i}Tr\big\{ \mathbf{\hat{M}}_{j}\boldsymbol{\hat{\rho}}_{i}\mathbf{\hat{M}}_{j}^{\dagger}\big\}    
\end{equation}
The term $Tr\big\{\mathbf{\hat{M}}_{i}\boldsymbol{\hat{\rho}}_{i}\mathbf{\hat{M}}_{i}^{\dagger}\big\} $ is the probability that the system modeled by the mixture (\ref{eqn:countablemixture}) was in the state $\boldsymbol{\hat{\rho}}_{i}$ given that the outcome of the measurement was the state $
\frac{\mathbf{\hat{M}}_{i}\Big(\sum_{i=1}^{N}p_{i} \boldsymbol{\hat{\rho}}_{i}\Big)\mathbf{\hat{M}}_{i}^{\dagger}}{Tr\big\{ {\mathbf{\hat{M}}_{i}\Big(\sum_{i=1}^{N}p_{i} \boldsymbol{\hat{\rho}}_{i}\Big) \mathbf{\hat{M}}_{i}^{\dagger}\big\}}}$, meaning that $\sum_{i=1}^{N}p_{i}Tr\big\{\mathbf{\hat{M}}_{i}\boldsymbol{\hat{\rho}}_{i}\mathbf{\hat{M}}_{i}^{\dagger}\big\}$ is the probability that the POVM \emph{discriminates} between the different $\boldsymbol{\hat{\rho}}_{i}$ of the mixture (\ref{eqn:countablemixture}). The term $1-\sum_{i=1}^{N}p_{i}Tr\big\{\mathbf{\hat{M}}_{i}\boldsymbol{\hat{\rho}}_{i}\mathbf{\hat{M}}_{i}^{\dagger}\big\}$, the \emph{probability error}, is hence the probability that the POVM fails to discriminate between the elements of the mixture (\ref{eqn:countablemixture}). Using (\ref{eqn:cuppis3}) it may also be expressed as $\sum_{j=1}^{N}\sum_{i; i\neq j}^{N}p_{i}Tr\big\{ \mathbf{\hat{M}}_{j}\boldsymbol{\hat{\rho}}_{i}\mathbf{\hat{M}}_{j}^{\dagger}\big\}$. In what follows we will just write $p_{E}$ in place of $p_{E}\big\{\{p_{i},\boldsymbol{\hat{\rho}}_{i}\}_{i=1}^{N}, \{\mathbf{\hat{M}}_{l}\big\}_{l=1}^{K} \big\}$when the context is clear. With this notation, the QSD optimization problem is the problem of computing the following minimum. 
\begin{equation}
\label{eqn:minerror3}
\min_{POVM}p_{E}
\end{equation}
The QSD problem will be called \emph{fully solvable} when $\min_{POVM}p_{E} = 0 $.\\

In this paper we consider an {\it Asymptotic} QSD: let $\mathscr{E}_{i}$ be a family of completely positive linear transformations, mapping density operators $\boldsymbol{\hat{\rho}}$ acting in a Hilbert space $\mathscr{H}_{1}$ to density operators $\mathscr{E}_{i}\big(\boldsymbol{\hat{\rho}}\big)$ acting in a Hilbert space $\mathscr{H}_{2}$ (equal to  $\mathscr{H}_{1}$ or not).  These operators depend on parameters $n_{i}$, and we consider the QSD problem for the mixture of the operators $\mathscr{E}_{i}\big(\boldsymbol{\hat{\rho}}\big)$ The object of our interest is the asymptotic behavior of the minimal error (\ref{eqn:minerror3}) corresponding to this QSD, in the limit as some or all of the parameters in $\{n_{i}\}_{i}$ go to infinity. We will say that the asymptotic QSD problem is \emph{fully solvable} with respect to the parameter $n_{i}$ when 
\begin{equation}
\label{eqn:minerror200}
\lim_{|i|\rightarrow \infty}\min_{POVM}p_{E}\big\{\{p_{i},\mathscr{E}_{i}\big(\boldsymbol{\hat{\rho}}\big)\}_{i=1}^{N}, \{\mathbf{\hat{M}}_{l}\big\}_{l=1}^{K} \big\} = 0
\end{equation}
the minimization above is understood to be taken for every $i$.

Asymptotic QSD arises naturally in the study of quantum communication, quantum-to-classical transition, and quantum measurement, to name a few applications \cite{szkola}\cite{JKthree}\cite{bae} \cite{cohen}. As an example, consider the case where a state is redundantly prepared by some party $A$ in the state $\boldsymbol{\hat{\rho}}_{i}$  with probability $p_{i}$, $n$ copies of each state being made prior to being communicated to another party $B$. From the perspective of $B$, the received state is a mixture of the form 
\begin{equation}
\label{eqn:spoon}
\sum_{i}p_{i}\boldsymbol{\hat{\rho}}_{i}^{\otimes n}
\end{equation}
In such case the map $\mathscr{E}_{n}$ maps operators $\boldsymbol{\hat{\rho}}$ to $\boldsymbol{\hat{\rho}}^{\otimes n}$.
Now, define 
\begin{equation}
\min_{POVM}p_{E}(n):= \min_{POVM}p_{E}\big\{\{p_{i},\boldsymbol{\hat{\rho}}_{i}^{\otimes n}\}_{i=1}^{N}, \{\mathbf{\hat{M}}_{l}\big\}_{l=1}^{K} \Big\}.
\end{equation}
In \cite{szkola} it was shown that 
\begin{equation}
\label{eqn:cofee10}
\frac{1}{3}C\leq -\lim_{n\rightarrow \infty} \frac{\log\Big(\min_{POVMP}p_{E}(n)\Big)}{n} \leq C
\end{equation}
where $C$ is a constant involving the quantum Chernoff bound for a mixture of $N$ states; 
a more detailed discussion may be found in appendix \ref{eqn:cherr}.
Hence, for large enough $n$, from eqn (\ref{eqn:cofee10}) we have the following inequalities. 
\begin{equation}
\label{eqn:cofee101}
e^{-n\frac{1}{3}C}\geq\min_{POVM}p_{E}(n) \geq  e^{-nC}
\end{equation}
This shows that the minimum error probability drops off exponentially as the redundancy $n$ grows.\\

%One may also use a Quantum Chernoff bound free method for the computation of $\lim_{n\rightarrow\infty}\min_{POVM}p_{E}(n)$ by applying a bound to the minimal probability error in \cite{knill} (Theorem \ref{eqn:knill1} in the following section). Let $\big\|\mathbf{\hat{A}}\big\|_{1}:= Tr\big\{\sqrt{\mathbf{\hat{A}}^{\dagger}\mathbf{\hat{A}}}\big\}$ be the trace norm of the operator $\mathbf{\hat{A}}$ Applying Theorem \ref{eqn:knill1}  we have the following result. 
%\begin{equation}
%\min_{POVM}p_{E}(n)\leq \sum_{i}\sum_{j;j\neq i}\sqrt{p_{i}p_{j}}\Big\|\sqrt{\boldsymbol{\hat{\rho}}_{i}^{\otimes n}}\sqrt{\boldsymbol{\hat{\rho}}_{j}^{\otimes n}}\Big\|_{1} = \sum_{i}\sum_{j;j\neq i}\sqrt{p_{i}p_{j}}\Big\|\sqrt{\boldsymbol{\hat{\rho}}_{i}}\sqrt{\boldsymbol{\hat{\rho}}_{j}}\Big\|_{1}^{n}
%\end{equation}
%which decays to zero as $n\rightarrow \infty$ when $\Big\|\sqrt{\boldsymbol{\hat{\rho}}_{i}}\sqrt{\boldsymbol{\hat{\rho}}_{j}}\Big\|_{1}<1$ for all $i,j; j\neq i$. 
%One of the remarkable aspects of such a result is the state-independent nature of the convergence, i.e. so long as the fidelity condition stated in the previous sentence is satisfied, the type of states are irrelevant; e.g. they can be finite or infinite dimensional density operators. QSD will nevertheless be dependent on the states $\boldsymbol{\hat{\rho}}_{i}$ constituting the mixture in general. Our next example exemplifies this.\\

More recently, and more pertinently to the theme of this paper, asymptotic QSD has made an appearance in the theory of Spectrum Broadcast Structures (SBS) \cite{JKthree} for \emph{quantum measurement limit}  type interactions  (see section 2.4 in \cite{schloss2} for a discussion of \emph{quantum measurement limit}). The SBS framework aims at deriving a specific type of state from the dynamics. These states are called SBS states; they satisfy a notion of objectivity presented in \cite{JKthree} \cite{JKtwo}  \cite{JKone} and emerge from the asymptotic dynamics studied therein. The definition of an SBS state stipulates a calculation related to QSD when proving that a state of interest is asymptotically an SBS state:  using the fact that $\big|Tr\big\{\mathbf{\hat{A}}\big\}\big|\leq \big\|\mathbf{\hat{A}}\big\|_{1}$ for any trace class operator $\mathbf{\hat{A}}$, we have

\begin{equation}
\label{eqn:burr}
\min_{POVM}\sum_{i=1}^{N}p_{i}\Big\| \boldsymbol{\hat{\rho}}_{i}-\mathbf{\hat{M}}_{i}\boldsymbol{\hat{\rho}}_{i}\mathbf{\hat{M}}^{\dagger}_{i} \Big\|_{1}\geq 
\end{equation}
\begin{equation}
\min_{POVM}\sum_{i=1}^{N}p_{i}\Big|Tr\Big\{ \boldsymbol{\hat{\rho}}_{i}-\mathbf{\hat{M}}_{i}\boldsymbol{\hat{\rho}}_{i}\mathbf{\hat{M}}^{\dagger}_{i} \Big\}\Big|\geq    
\end{equation}
\begin{equation}
\min_{POVM}\Big( 1-\sum_{i=1}^{N}p_{i}Tr\big\{\mathbf{\hat{M}}_{i}\boldsymbol{\hat{\rho}}_{i}\mathbf{\hat{M}}^{\dagger}_{i} \big\}  \Big) 
\end{equation}
which is just the QSD error bound. 
In \cite{JKone} \cite{JKtwo} \cite{JKthree}, special attention has been given to optimization problems of the following form. 
\begin{equation}
\label{cofee1000}
 \min_{POVM}\sum_{i}p_{i}\Big\|e^{-itx_{i}\mathbf{\hat{B}}}\boldsymbol{\hat{\rho}}e^{itx_{i}\mathbf{\hat{B}}}-\mathbf{\hat{M}}_{i}e^{-itx_{i}\mathbf{\hat{B}}}\boldsymbol{\hat{\rho}}e^{itx_{i}\mathbf{\hat{B}}}\mathbf{\hat{M}}_{i}^{\dagger}\Big\|_{1}   
\end{equation}
where $x_{i}\neq x_{j}$ for all $i,j; j\neq i$ and $\mathbf{\hat{B}}$ is an arbitrary self-adjoint operator; the state being discriminated here is of course $\sum_{i}p_{i}e^{-itx_{i}\mathbf{\hat{B}}}\boldsymbol{\hat{\rho}}e^{itx_{i}\mathbf{\hat{B}}}$. Such unitarily related mixtures, with a parameter $t$, arise as a direct consequence of the aforementioned \emph{quantum measurement} limit assumption made in \cite{JKone}\cite{JKtwo}\cite{JKthree}. The modified QSD problem (\ref{cofee1000}) is asymptotically \emph{fully solvable} with respect to $t$ only if the original QSD problem is \emph{fully solvable} (in which the trace norm is replaced by the trace). Unlike inexample (\ref{eqn:spoon}), where the asymptotic \emph{full solvability} of the respective QSD problem was independent of the nature of the states involved, here this is not the case. It is easy to find examples where (\ref{cofee1000}) does not vanish as $t\rightarrow \infty$. E.g. let $\mathbf{\hat{B}}$ be equal to the Pauli matrix  $\boldsymbol{\hat{\sigma}}_{x}$ and let $\boldsymbol{\hat{\rho}} = \big|z_{1}\big\rangle\big\langle z_{1}\big|$ with $\big\{\big|z_{i}\big\rangle \big\}_{i}$ the eigenvectors of the Pauli matrix $\boldsymbol{\hat{\sigma}}_{z}$. After some calculation it can be shown that
\begin{equation}
 e^{-itx_{i} \boldsymbol{\sigma}_{x}}\boldsymbol{\hat{\rho}}e^{itx_{i}\boldsymbol{\sigma}_{x}} =  
\end{equation}
\begin{equation}
 \cos^{2}(tx_{i})\big|z_{1}\big\rangle\big\langle z_{1}\big|+\sin^{2}(tx_{i})\big|z_{2}\big\rangle\big\langle z_{2}\big|+
 \end{equation}
 \begin{equation}
 i\cos(tx_{i})\sin(tx_{i})\big|z_{1}\big\rangle\big\langle z_{2}\big|-i\cos(tx_{i})\sin(tx_{i})\big|z_{2}\big\rangle\big\langle z_{1}\big|
\end{equation}
Now consider the mixture
\begin{equation}
\label{eqn:hallo}
\sum_{i=1}^{2}p_{i}e^{-itx_{i} \boldsymbol{\sigma}_{x}}\boldsymbol{\hat{\rho}}e^{itx_{i}\boldsymbol{\sigma}_{x}}
\end{equation}
and let $p_{i}=\frac{1}{2}$ for $i = 1, 2$. An application of a result by Hellstr\"om \cite{hellstrom}, discussed in the next section leads to  
\begin{equation}
\label{eqn:periodic}
\min_{POVM}p_{E}(t) = \frac{1}{2}-\frac{1}{4}\big\|e^{-itx_{1} \boldsymbol{\sigma}_{x}}\boldsymbol{\hat{\rho}}e^{itx_{1}\boldsymbol{\sigma}_{x}}-e^{-itx_{1} \boldsymbol{\sigma}_{x}}\boldsymbol{\hat{\rho}}e^{itx_{1}\boldsymbol{\sigma}_{x}} \big\|_{1}=
\end{equation}
\begin{equation}
2\Big|\sqrt{\big(\cos^{2}(tx_{1})-\cos^{2}(tx_{2})\big)\big(\sin^{2}(tx_{1})-\sin^{2}(tx_{2})\big)-\big(\big|\cos(tx_{1})\sin(tx_{1})-\cos(tx_{2})\sin(tx_{2})\big|^{2}\big)}\Big|
\end{equation}
It is clear that (\ref{eqn:periodic}) does not converge to zero as $t\rightarrow \infty$, so asymptotic QSD is not \emph{fully solvable} and, as a consequence of (\ref{eqn:burr}), the value (\ref{cofee1000}) does not converge to $0$ either.\\

In this paper, we will be focusing on the QSD of unitarily related mixtures (URM); i.e. mixtures of the form $\sum_{i}p_{i}\mathbf{\hat{U}}_{i}(t)\boldsymbol{\hat{\rho}}\mathbf{\hat{U}}_{i}^{\dagger}(t)$, where $\mathbf{\hat{U}}_{i}(t)$ is a one-parameter group of unitary operators. We will provide a necessary and sufficient condition for the asymptotic \emph{full solvability} of the QSD optimization problem for a broad set of URM; this condition will be expressed in terms of the spectral properties of the generator of the unitary group characterizing the URM and the nature of the initial state. In Sections 2,  and 3 we will give an overview of some important results from the literature that we shall be using and give further motivation. In section 4 we present the main results (Theorem \ref{eqn:maintheorem2} and Corollary \ref{eqn:mainresult3}) which give necessary and sufficient conditions for asymptotic QSD optimization of unitarily related mixtures to be \emph{fully solvable} in a broad setting. Drawing parallels between QSD and UQSD we prove a necessary condition for UQSD in the unitarily related mixture case to be  \emph{fully solvable} in the infinite time limit. This condition will again depend only on the spectral properties of the generator of the unitary group characterizing the URM and on the nature of the initial state. We conclude by conjecturing that an analog of Theorem \ref{eqn:maintheorem2} is true for the UQSD case in the unitarily related mixture setting; we follow this conjecture with some motivation and intuition. 

\section{Some Important Theorems}
 \;\;\; The optimization problem (\ref{eqn:minerror3}) is, in general, intractable; exact solutions exist only in a few cases \cite{bae} \cite{barnett}. The most famous of them is the already mentioned so-called \emph{Hellstr\"{o}m bound}. The name may be misleading because it is not a bound. We present its statement following \cite{bae}. 
\begin{definition4}{\textbf{Hellstr\"{o}m Bound}:}
\label{eqn:hellstrom}
Let $\mathscr{H}$ be an arbitrary $Hilbert$ space. For any mixture of the form 
\begin{equation}
p_{1}\boldsymbol{\hat{\rho}}_{1}+p_{2}\boldsymbol{\hat{\rho}}_{2} \in \mathcal{S}(\mathscr{H})
\end{equation}
we have 
\begin{equation}
\min_{POVM}p_{E}\big\{\{p_{i},\boldsymbol{\hat{\rho}}_{i}\}_{i=1}^{2}, \{\mathbf{\hat{M}}_{l}\big\}_{l=1}^{2} \big\} = \frac{1}{2}-\frac{1}{2}\big\|p_{1}\boldsymbol{\hat{\rho}}_{1}- p_{2}\boldsymbol{\hat{\rho}}_{2}\big\|_{1}.
\end{equation}
Here $\big\|\mathbf{\hat{A}}\big\|_{1}:=Tr\big\{\sqrt{\mathbf{\hat{A}}^{\dagger}\mathbf{\hat{A}}}\big\}$ (Trace norm). 
\end{definition4}

Lower and upper bounds for the probability error in the case of a general mixture exist. Some of the more famous are the following. Let the $\big\{\boldsymbol{\hat{\rho}}_{i}\big\}_{i=1}^{N}$ be density operators acting in a Hilbert space $\mathscr{H}$ and $\{p_{i}\}_{i=1}^{N}$ a probability distribution, then 
\begin{definition4}{\textbf{Qiu Bound} \cite{qiu}}
\label{eqn:qiu1}
\begin{equation}
\label{eqn:qiu}
\min_{POVM}p_{E}\geq\frac{1}{2}\big(1-\frac{1}{2(N-1)}\sum_{i}\sum_{j;j\neq i}\big\|p_{i}\boldsymbol{\hat{\rho}}_{i}-p_{j}\boldsymbol{\hat{\rho}}_{j}\big\|_{1}\big)
\end{equation}
\end{definition4}
\begin{definition4}{\textbf{Montanaro Bound} \cite{Montanaro}:}
\label{eqn:monta}
\begin{equation}
\min_{POVM}p_{E}\geq \frac{1}{2}\sum_{i}\sum_{j;j\neq i}p_{i}p_{j}F(\boldsymbol{\hat{\rho}}_{i},\boldsymbol{\hat{\rho}}_{j})
\end{equation}
\end{definition4}
\begin{definition4}{\textbf{Knill and Barnum} \cite{knill}}
\label{eqn:knill1}
\begin{equation}
\label{eqn:knill}
  \min_{POVM}p_{E} \leq \sum_{i}\sum_{j;j\neq i}\sqrt{p_{i}p_{j}}\sqrt{F\big(\boldsymbol{\hat{\rho}}_{i} , \boldsymbol{\hat{\rho}}_{j}\big)}  
\end{equation} 
Here $F\big(\boldsymbol{\hat{\rho}},\boldsymbol{\hat{\sigma}}\big):=\big\|\sqrt{\boldsymbol{\hat{\rho}}}\sqrt{\boldsymbol{\hat{\sigma}}}\big\|_{1}^{2}$.
\end{definition4}
Theorem \ref{eqn:knill1} is proven for the case of finite-dimensional Hilbert spaces. We provide a generalization of this theorem for the case where $\mathscr{H}$ is infinite-dimensional in appendix A.\\

In \cite{qiu}, necessary and sufficient conditions are introduced in order to arrive at a generalization of the Hellstr\"{o}m bound. Unlike Theorem 1, the Qiu bound does not provide explicit knowledge of the POVM that minimizes $p_{E}$. There also exist convex optimization techniques that may be employed in order to find a global minimum \cite{bae}\cite{bae2}
, but we will not be pursuing these ideas.\\

We now present some results particular to the quantum fidelity. 
\begin{definition4}{\textbf{Purification-dependent version of the fidelity}:}
\label{eqn:maxyo}
The quantum fidelity $F\big(\boldsymbol{\hat{\rho}},\boldsymbol{\hat{\sigma}}\big):=\Big\|\sqrt{\boldsymbol{\hat{\rho}}}\sqrt{\boldsymbol{\hat{\sigma}}}\Big\|_{1}^{2} = Tr\Big\{\sqrt{\sqrt{\boldsymbol{\hat{\rho}}}\boldsymbol{\hat{\sigma}}\sqrt{\boldsymbol{\hat{\rho}}}}\Big\}^{2}$ is equivalent to the following \cite{Nielsen}. 
\begin{equation}
F\big(\boldsymbol{\hat{\rho}},\boldsymbol{\hat{\sigma}}\big)=
\max_{|\chi\rangle}\big|\big\langle \xi\big|\chi\big\rangle\big|^{2}
\end{equation}
where $\big|\psi\big\rangle$ is any fixed purification of $\boldsymbol{\hat{\rho}}$, and the maximization is over all purifications of $\boldsymbol{\hat{\sigma}}$.
\end{definition4}

\begin{definition2}{\textbf{Strong concavity of the fidelity; a theorem analogous to the similarly titled theorem about discrete distributions in} \cite{Nielsen}: }
\label{eqn:subcont}
Let $\int p(x) \boldsymbol{\hat{\rho}}_{x}dx$ and $\int q(x) \boldsymbol{\hat{\sigma}}_{x}dx$ be two uncountable mixtures (p(x) and q(x) are probability distributions). Then, 

\begin{equation}
 \sqrt{F\Big(\int p(x) \boldsymbol{\hat{\rho}}_{x}dx, \int q(x) \boldsymbol{\hat{\sigma}}_{x}dx\Big)}\geq \int \sqrt{p(x)q(x)}\sqrt{F\big( \boldsymbol{\hat{\rho}}_{x}, \boldsymbol{\hat{\sigma}}_{x}\big)}dx   
\end{equation}
\end{definition2}
\begin{proof}
The proof follows the standard methodology,  see chapter 9 of \cite{Nielsen} for the countable mixture case. Let $\big|\psi_{x}\big\rangle$ and $\big|\sigma_{x}\big\rangle$ be the purifications of $\boldsymbol{\hat{\rho}}_{x}$ and $\boldsymbol{\hat{\sigma}}_{x}$ which maximize the fidelity; i.e. $F\Big( \boldsymbol{\hat{\rho}}_{x}, \boldsymbol{\hat{\sigma}}_{x}\Big)= \big|\big\langle \psi_{x}\big|\phi_{x}\big\rangle\big|^{2}$. We now define
\begin{equation}
\big|\psi\big\rangle:=\int \sqrt{p(x)}\big|\psi_{x}\big\rangle\big|x\big\rangle dx
\end{equation}
\begin{equation}
\big|\phi\big\rangle:=\int \sqrt{q(x)}\big|\phi_{x}\big\rangle\big|x\big\rangle dx.
\end{equation}
$\big|\psi\big\rangle$ and $\big|\phi\big\rangle$ are purifications of the operators $\int p(x) \boldsymbol{\hat{\rho}}_{x}dx$ and $\int q(x) \boldsymbol{\hat{\sigma}}_{x}dx$ where the ancillary space is taken to be $L^{2}\big(\mathbb{R}\big)$. Using Theorem \ref{eqn:maxyo} we have. 
\begin{equation}
\sqrt{F\Big(\int p(x) \boldsymbol{\hat{\rho}}_{x}dx, \int q(x) \boldsymbol{\hat{\sigma}}_{x}dx\Big)}\geq |\big\langle \phi\big|\psi\big\rangle\big| = 
\end{equation}
\begin{equation}
\Bigg|\int \sqrt{p(x)}\sqrt{q(y)}\big\langle\psi_{x}\big|\phi_{y}\big\rangle \big\langle x \big| y\big\rangle dydx \Bigg| =  \Bigg|\int\int \sqrt{p(x)q(x)}\big\langle\psi_{x}\big|\phi_{x}\big\rangle dx \Bigg| =  
\end{equation}
\begin{equation}
\Bigg|\int \sqrt{p(x)q(x)}\big\langle\psi_{x}\big|\phi_{x}\big\rangle dx \Bigg| =\int \sqrt{p(x)q(x)}\sqrt{F\big(\boldsymbol{\hat{\rho}}_{x},\boldsymbol{\hat{\sigma}}_{x}\big)}dx 
\end{equation}
\end{proof}
The latter gives us the means by which we may bound fidelities of mixed state from below. The following corollary follows immediately from Lemma \ref{eqn:subcont}. 
\begin{Co}
\label{eqn:explion}
Let $\int p(x) \boldsymbol{\hat{\rho}}_{x}dx$ be an uncountable mixture, and let $ \boldsymbol{\hat{\sigma}}$ be an arbitrary density operator (p(x) is a probability distribution). Then, 
\begin{equation}
 \sqrt{F\Big(\int p(x) \boldsymbol{\hat{\rho}}_{x}dx, \boldsymbol{\hat{\sigma}}\Big)}\geq \int p(x)\sqrt{F\big( \boldsymbol{\hat{\rho}}_{x}, \boldsymbol{\hat{\sigma}}\big)}dx   
\end{equation}
\end{Co}
\begin{proof}
 Note that $\boldsymbol{\hat{\sigma}} = \int p(x)\boldsymbol{\hat{\sigma}} dx $. The proof follows from applying Lemma (\ref{eqn:subcont}) to the fidelity $\sqrt{F\Big(\int p(x) \boldsymbol{\hat{\rho}}_{x}dx, \int p(x) \boldsymbol{\hat{\sigma}}dx\Big)}$. 
\end{proof}

\section{Countable Mixtures of Unitarily Related Families.}
\;\;\; In this section, we restrict our attention to a specific type of ensemble $\{p_{i},\boldsymbol{\hat{\rho}}_{i,t}\}_{i=1}^{N}$. Namely, let 
\begin{equation}
\label{eqn:thestatesoyeah}
\boldsymbol{\hat{\rho}}_{i,t} := e^{-itx_{i}\mathbf{\hat{B}}}\big|\psi\big\rangle\big\langle \psi\big|e^{itx_{i}\mathbf{\hat{B}}}
\end{equation}
for some self-adjoint operator $\mathbf{\hat{B}}$ and a pure density operator $\big|\psi\big\rangle\big\langle \psi\big|$ both acting in a Hilbert space $\mathscr{H}$. The asymptotic discriminability of the mixture $\sum_{i=1}^{N}p_{i}\boldsymbol{\hat{\rho}}_{i,t}$ will be shown to depend on the spectral properties of the operators $\mathbf{\hat{B}}_{k}$ and on the nature of the pure state $\big|\psi\big\rangle\big\langle \psi\big|$. Using Theorem \ref{eqn:knill1} we have the following QSD estimate. 
\begin{equation}
\label{eqn:boigacrey}
\min_{POVM}\Big(1-\sum_{i=1}^{N}p_{i}Tr\big\{\mathbf{\hat{M}}_{i}\boldsymbol{\hat{\rho}}_{i,t}\mathbf{\hat{M}}^{\dagger}_{i} \big\}\Big)\leq \sum_{i}\sum_{j;j\neq i}\sqrt{p_{i}p_{j}}\sqrt{F\big(\boldsymbol{\hat{\rho}}_{i,t},\boldsymbol{\hat{\rho}}_{j,t}\big)}=
\end{equation}
\begin{equation}
\label{eqn:pc}
\sum_{i}\sum_{j;j\neq i}\sqrt{p_{i}p_{j}}\big|\big\langle \psi\big|e^{-it(x_{j}-x_{i})\mathbf{\hat{B}}}\big|\psi\big\rangle\big|\end{equation} 

If \emph{fully solvable} asymptotic QSD is desired in such a case, then it is a necessary and sufficient condition that the elements $\big|\big\langle \psi\big|e^{-it(x_{j}-x_{i})\mathbf{\hat{B}}}\big|\psi\big\rangle\big|$ of the sum above decay to zero as $t\rightarrow \infty$ for all $i,j;i\neq j$. The necessity follows from applying the bound in Theorem \ref{eqn:monta} to the left-hand side of (\ref{eqn:boigacrey}). Sufficiency can be proven by employing Theorem \ref{eqn:knill1}. In \ref{eqn:zelda} of this section we prove necessary and sufficient conditions for full solvability of QSD for a particular set of URM.  
\subsection{Spectral Decomposition and Spectral Measures.}
\;\;\; We start from some spectral theory background \cite{Simon}\cite{pascual}.  Let $\mathbf{\hat{A}}$ be a self-adjoint operator acting in a Hilbert space $\mathscr{H}$.  We have
\begin{equation}
\label{eqn:spec3}
\mathbf{Spec}\big(\mathbf{\hat{A}}\big) = \mathbf{Spec}_{p}\big(\mathbf{\hat{A}}\big)\cup \mathbf{Spec}_{ac}\big(\mathbf{\hat{A}}\big)\cup\mathbf{Spec}_{sc}\big(\mathbf{\hat{A}}\big)\
\end{equation}
where the subscripts $ac$ and $sc$ stand for absolutely continuous and singular continuous respectively. To formally define absolutely continuous and singular continuous spectra let us consider an arbitrary $\big|\psi\big\rangle \in \mathscr{H}$ It follows from spectral theory that there exists a unique measure $\mu_{\psi}$ such that \cite{Simon} 
\begin{equation}
\big\langle \psi\big|\mathbf{\hat{A}}\big|\psi\big\rangle = \int_{\mathbb{R}}\lambda d\mu_{\psi}(\lambda).
\end{equation}
The measure $\mu_{\psi}$ is called the spectral measure associated with $\big|\psi\big\rangle$. By the \emph{Lebesgue Decomposition Theorem} any such measure is a sum of a point measure, an absolutely continuous measure, and a singular continuous measure: 
\begin{equation}
\label{eqn:mesdec}
\mu_{\psi}= \mu_{\psi, p}+\mu_{\psi,ac}+\mu_{\psi, sc}
\end{equation}
This decomposition is unique.
Of particular interest to us will be the properties of the Fourier transforms of the measures on the right-hand side of (\ref{eqn:mesdec}). It is a consequence of the Riemann-Lebesgue Lemma that the Fourier transform of the measure $\mu_{\psi,ac}$ (absolutely continuous with respect to the Lebesgue measure) is a function that decays to zero at infinity. On the other hand, it can be shown that the Fourier transform of $\mu_{\psi, p}$ will exhibit quasiperiodic behavior. In general, the Fourier transform of $\mu_{\psi,sc}$ does not decay to zero, but we shall be particularly interested in the subset of measures continuous with respect to the Lebesgue measures whose Fourier transforms have this property, since then we will be able to prove that the bound (\ref{eqn:knill1}) converges to zero for the case of the mixture presented in (\ref{eqn:thestatesoyeah}), i.e. $\sum_{i=1}^{N}p_{i}\boldsymbol{\hat{\rho}}_{i,t}$. \\

$\mathscr{H}$ may furthermore be expressed as a direct sum of three invariant subspaces; one corresponding to each type of spectrum. Namely, (see \cite{pascual})
\begin{equation}
\mathscr{H} = \mathscr{H}_{p}\oplus \mathscr{H}_{ac}\oplus \mathscr{H}_{sc}
\end{equation}
Recall that QSD may be \emph{ fully solved} asymptotically iff $\forall i\neq j$
\begin{equation}
\label{eqn:link}
\big|\big\langle \psi\big|e^{-it(x_{j}-x_{i})\mathbf{\hat{B}}}\big|\psi\big\rangle\big|\rightarrow 0 \;\;\; (as \;\;\; t\rightarrow \infty)
\end{equation}
It follows from the spectral theorem for unitary operators \cite{Simon} that (\ref{eqn:link}) may be written as 
\begin{equation}
\big|\big\langle \psi\big|e^{-it(x_{j}-x_{i})\mathbf{\hat{B}}}\big|\psi\big\rangle\big| = \Big|\int_{\mathbb{R}}e^{-it(x_{j}-x_{i})\lambda}d\mu_{\psi}(\lambda)\Big| \leq 
\end{equation}
\begin{equation}
\Big|\int_{\mathbb{R}}e^{-it(x_{j}-x_{i})\lambda}d\mu_{\psi,p}(\lambda)\Big|   + \Big|\int_{\mathbb{R}}e^{-it(x_{j}-x_{i})\lambda}d\mu_{\psi,ac}(\lambda)\Big| + \Big|\int_{\mathbb{R}}e^{-it(x_{j}-x_{i})\lambda}d\mu_{\psi,sc}(\lambda)\Big| 
\end{equation}
For $\big|\big\langle \psi\big|e^{-it(x_{j}-x_{i})\mathbf{\hat{B}}}\big|\psi\big\rangle\big|$ to go to $ 0 $ as $t\rightarrow \infty$, it is necessary that $\big| \psi\big\rangle \in \mathscr{H}_{ac}\oplus\mathscr{H}_{sc}$ \cite{gerald}. We must hence constrain ourselves further to the subspace $\mathscr{H}_{rc}$ consisting only of the states $\big|\psi\big\rangle$ whose associated measure $\mu_{\psi}$ is a \emph{Rajchman} measure \cite{gerald} (defined below).
\begin{definition}{\textbf{Rajchman Measure}:}
A finite Borel probability measure $\mu$ on $\mathbf{R}$ is called a Rajchman measure if it satisfies 
\begin{equation}
\lim_{t\rightarrow\infty}\hat{\mu}(t) = 0 
\end{equation}
where
$\hat{\mu}(t) := \int_{\mathbb{R}} e^{2i\pi tx}d\mu(x)$ \;\;, $t\in \mathbb{R}$.
\end{definition}

\begin{definition4}
\label{eqn:firstresult}
Let $\mathbf{\hat{A}}$ be a self-adjoint operator acting on some arbitrary Hilbert space $\mathscr{H}$, then the set of vectors in $\mathscr{H}$ for which the spectral measure is a Rajchman measure, i.e.
\begin{equation}
\label{eqn:link2}
 \mathscr{H}_{rc}:=\Big\{\big|\psi\big\rangle\;| \;\lim_{t\rightarrow \infty} \big\langle \psi\big|e^{-it\mathbf{\hat{A}}}\big|\psi\big\rangle = 0\big\},
\end{equation}
is a closed subspace which is invariant under $e^{-is\mathbf{\hat{A}}}$ \cite{gerald}.
\end{definition4}
In what follows, $\mu_{\phi,\psi}$ is defined implicitly. Given a unitary operator $e^{-it\mathbf{\hat{B}}}$, we compute $\big\langle \phi\big|e^{-it\mathbf{\hat{B}}}\big|\psi\big\rangle = \big\langle \phi\big|\int e^{-it\lambda}d\mathbf{\hat{E}}_{\lambda}\big|\psi\big\rangle = \int e^{-it\lambda}d\big\langle \phi\big|\mathbf{\hat{E}}_{\lambda}\big|\psi\big\rangle $. In the second to last term, we utilized the spectral decomposition of $\mathbf{\hat{B}}$; it is $\big\langle \phi\big|\mathbf{\hat{E}}_{\lambda}\big|\psi\big\rangle$ that we will refer to as $\mu_{\phi,\psi}$.

\begin{definition2}{} 
\label{eqn:rajj}
Let $\mathbf{\hat{B}}$ be a self-adjoint operator acting in a Hilbert space $\mathscr{H}$. Furthermore, let $\big|\psi\big\rangle \in \mathscr{H}_{rc}$ and $\big|\phi\big\rangle \in \mathscr{H}$, then the measure $\mu_{\phi,\psi}$ is Rajchman. 
\end{definition2}
\begin{proof}
\begin{equation}
\int e^{-it\lambda}
\mu_{\phi,\psi}(\lambda) = \big\langle \phi\big|e^{-it\mathbf{\hat{B}}}\big|\psi\big\rangle = \big\langle \phi\big|\Big(\mathbf{\hat{P}}_{rc}e^{-it\mathbf{\hat{B}}}\big|\psi\big\rangle\Big) = 
\end{equation}
\begin{equation}
 \Big(\big\langle \phi\big|\mathbf{\hat{P}}_{rc}\Big)e^{-it\lambda}\big|\psi\big\rangle =\big \langle \xi\big|e^{-it\mathbf{\hat{B}}}\big|\psi\big\rangle  
\end{equation}
where $\mathbf{\hat{P}}_{rc}$ is the projector onto the subspace $\mathscr{H}_{rc}$ and $\big|\xi\big\rangle := \mathbf{\hat{P}}_{rc}\big|\phi\big\rangle\in \mathscr{H}_{rc}$. We have used the fact that the Rajchman subspace is invariant under the action of $e^{-it\mathbf{\hat{B}}}$. Using the polarization identity (see \cite{giacomo}, Chapter 2, Excersise 2.1) we have
\begin{equation}
\big \langle \xi\big|e^{-it\mathbf{\hat{B}}}\big|\psi\big\rangle = 
\frac{1}{4}\sum_{k=0}^{3}i^{k}\Bigg(\Big(\big\langle \xi\big|+(-i)^{k}\big\langle \psi\big|\Big)e^{-it\mathbf{\hat{B}}}\Big(\big|\xi\big\rangle+i^{k}\big|\psi\big\rangle\Big)\Bigg) = 
\end{equation} 
\begin{equation}
\frac{1}{4}\sum_{k=0}^{3}i^{k}\big\langle \chi_{k}\big|e^{-it\mathbf{\hat{B}}}\Big|\chi_{k}\big\rangle
\end{equation}
where we have defined $\big|\chi_{k}\big\rangle:\big|\xi\big\rangle+i^{k}\big|\psi\big\rangle$. $\mathscr{H}_{rc}$ is a linear space, hence $\big|\chi_{k}\big\rangle\in\mathscr{H}_{rc}$ for $k = 0,1,2,3$. It follows that
\begin{equation}
\int e^{-it\lambda}
\mu_{\phi,\psi}(\lambda) =  \frac{1}{4}\sum_{k=0}^{3}i^{k}\big\langle \chi_{k}\big|e^{-it\mathbf{\hat{B}}}\Big|\chi_{k}\big\rangle = \frac{1}{4}\sum_{k=0}^{3}i^{k}\int e^{-it\lambda}d\mu_{\chi_{k}}(\lambda).  
\end{equation}
As $t\rightarrow \infty$ $\int e^{-it\lambda}d\mu_{\chi_{k}}(\lambda)\rightarrow 0$. Hence $\int e^{-it\lambda}
\mu_{\phi,\psi}(\lambda)\rightarrow 0$ as $t\rightarrow \infty$.
\end{proof}

We conclude this subsection with the following proposition.
\begin{definition9}{\textbf{Full Solvability of QSD for URM of Pure States}:}
\label{eqn:zelda}
Consider the model described in this section by the states (\ref{eqn:thestatesoyeah}).  $\big|\psi\big\rangle\in\mathscr{H}_{rc}$ corresponding to $\mathbf{\hat{B}}$ iff
\begin{equation}
 \lim_{t\rightarrow \infty} \min_{POVM}p_{E}\Big\{\big\{p_{i},\;e^{-itx_{i}\mathbf{\hat{B}}}\big|\psi\big\rangle\big\langle \psi\big|e^{itx_{i}\mathbf{\hat{B}}}\big\}_{i=1}^{N}, \{\mathbf{\hat{M}}_{l}\big\}_{l=1}^{K} \Big\} = 0  
\end{equation}
\end{definition9}

\begin{proof}
This immediately follows from Theorems \ref{eqn:knill1} and \ref{eqn:monta}.
\end{proof}

\section{Unitarily Related Linear Combinations of Finite Mixtures.}
\label{eqn:yoyobitte}
    Let us now consider the case where 
\begin{equation}
\label{eqn:countablemixture2}
\boldsymbol{\hat{\rho}} = \sum_{i=1}^{N}p_{i} \boldsymbol{\hat{\rho}}_{i} \in \mathcal{S}(\mathscr{H}),
\end{equation}
with 
\begin{equation}
\boldsymbol{\hat{\rho}}_{i} = \sum_{j=1}^{M_{i}}\eta_{ij}\big|\phi_{ij}\big\rangle \big\langle \phi_{ij}\big|
\end{equation}
with all of the $\big|\phi_{ij}\big\rangle \in \mathscr{H}$ of norm one and $\sum_{j}\eta_{ij} = 1$. In this case we may again use Theorem \ref{eqn:knill1}, obtaining
\begin{equation}
\label{eqn:prekoen}
 \min_{POVM}p_{E}(t)\leq \sum_{i}\sum_{j;j\neq i}\sqrt{p_{i}p_{j}}\sqrt{F\big(\boldsymbol{\hat{\rho}}_{i}, \boldsymbol{\hat{\rho}}_{j} \big)} 
\end{equation} 
However, the fidelities in this case are not trivial because both $\boldsymbol{\hat{\rho}}_{i}$ and $\boldsymbol{\hat{\rho}}_{j}$ are mixed states. To address this problem, we will use a bound for quantum fidelities from \cite{koenraad}. Namely,
\begin{definition4}{\textbf{Fidelity Bound (Koenraad and Milan} \cite{koenraad}\textbf{)}:} 
\label{eqn:koen}
Let $\sum_{i}p_{i}\boldsymbol{\hat{\rho}}_{i}$ be a countable mixture and let $\boldsymbol{\hat{\sigma}}$ be a density operator, acting on the same Hilbert space. Then, 
\begin{equation}
\sqrt{F\Big(\sum_{i}p_{i}\boldsymbol{\hat{\rho}}_{i}, \boldsymbol{\hat{\sigma}}\Big)}\leq \sum_{i}\sqrt{p_{i}}
\sqrt{F\big(\boldsymbol{\hat{\rho}}_{i},\boldsymbol{\hat{\sigma}}}\big)
\end{equation}
\end{definition4}

Applying Theorem \ref{eqn:koen} twice we may further bound (\ref{eqn:prekoen}) to obtain
\begin{equation}
 \min_{POVM}p_{E}\leq \sum_{i=1}^{N}\sum_{j;j\neq i}^{N}\sqrt{p_{i}p_{j}}\sum_{k=1}^{M_{i}}\sum_{k^{\prime}=1}^{M_{j}}\sqrt{\eta_{ik}}\sqrt{\eta_{ik^{\prime}}}\sqrt{F\big(\big|\phi_{ik}\big\rangle \big\langle \phi_{ik}\big|,\big|\phi_{jk^{\prime}}\big\rangle \big\langle \phi_{jk^{\prime}}\big|\big)} = 
\end{equation}
\begin{equation}
\label{eqn:keonraad3}
\sum_{i=1}^{N}\sum_{j;j\neq i}^{N}\sqrt{p_{i}p_{j}}\sum_{k=1}^{M_{i}}\sum_{k^{\prime}=1}^{M_{j}}\sqrt{\eta_{ik}\eta_{jk^{\prime}}}\big|\big\langle \phi_{ik}\big|\phi_{jk^{\prime}}\big\rangle\big|
\end{equation}
We thus see that the optimal probability error may be controlled by the inner products $\big|\big\langle \phi_{ik}\big|\phi_{jk}\big\rangle\big|$ $(i\neq j)$, which are relatively easy to compute. This leads to the following generalization of Proposition \ref{eqn:zelda}.

\begin{definition4}{\textbf{Full Solvability of QSD for URM of Finite Mixtures}:}
\label{eqn:maintheorem2}
Let $\mathscr{H}$ be an infinite-dimensional Hilbert space. Let $\mathbf{\hat{B}}$ be a self-adjoint operator acting in $\mathscr{H}$. Furthermore, let $\boldsymbol{\hat{\rho}}_{i}:=\sum_{j=1}^{M_{i}}\eta_{ij}\big|\phi_{ij}\big\rangle\big\langle\phi_{ij}\big|$ be finite mixtures in $\mathcal{S}\big(\mathscr{H}\big)$ for each $i$. Then, \begin{equation}
 \lim_{t\rightarrow \infty} \min_{POVM}p_{E}\Big\{\big\{p_{i},\;e^{-itx_{i}\mathbf{\hat{B}}}\boldsymbol{\hat{\rho}}_{i}e^{itx_{i}\mathbf{\hat{B}}}\big\}_{i=1}^{N}, \{\mathbf{\hat{M}}_{l}\big\}_{l=1}^{K} \Big\} = 0  
\end{equation}
iff all of the $\big|\phi_{ij}\big\rangle  \in \mathscr{H}_{rc}$ of $\mathbf{\hat{B}}$.
\end{definition4}
\begin{proof}
First we assume that $\big|\phi_{ij}\big\rangle  \in \mathscr{H}_{rc}$ of $\mathbf{\hat{B}}$ for all $ij$. Now, using (\ref{eqn:keonraad3}) we have 
\begin{equation}
\min_{POVM}p_{E}(t)\leq \sum_{i=1}^{N}\sum_{j;j\neq i}^{N}\sqrt{p_{i}p_{j}}\sum_{k=1}^{M_{i}}\sum_{k^{\prime}=1}^{M_{j}}\sqrt{\eta_{ik}\eta_{jk^{\prime}}}\big|\big\langle \phi_{ik}\big|e^{-it(x_{j}-x_{i})\mathbf{\hat{B}}}\big|\phi_{jk}\big\rangle\big|   
\end{equation}
Since all of the sums above are finite, we need only prove that the limits   
\begin{equation}
\lim_{t\rightarrow \infty}\big|\big\langle \phi_{ik}\big|e^{-it(x_{j}-x_{i})\mathbf{\hat{B}}}\big|\phi_{jk}\big\rangle\big|
\end{equation}
are zero.  This follows from Lemma \ref{eqn:rajj}.

To prove the converse, assume that for some $ij$ $\big|\phi_{ij}\big\rangle\notin \mathscr{H}_{rc}$. Using Theorem \ref{eqn:monta} we have 
\begin{equation}
\min_{POVM}p_{E}(t)\geq \frac{1}{2}\sum_{i=1}^{N}\sum_{j;j\neq i}^{N} p_{i}p_{j}F\Big(e^{-tx_{i}\mathbf{\hat{B}}}\boldsymbol{\hat{\rho}}_{i}e^{tx_{i}\mathbf{\hat{B}}},e^{-tx_{j}\mathbf{\hat{B}}}\boldsymbol{\hat{\rho}_{j}}e^{tx_{i}\mathbf{\hat{B}}}\Big)\geq
\end{equation}
\begin{equation}
\frac{1}{2}\sum_{i=1}^{N}\sum_{j;j\neq i}^{N} p_{i}p_{j}\Bigg(\sum_{k=1}^{\min\big\{M_{i},M_{j}\big\}}\sqrt{\eta_{ik}\eta_{jk}}\big|\big\langle \phi_{ik}\big|e^{-it(x_{j}-x_{i})\mathbf{\hat{B}}}\big|\phi_{jk}\big\rangle\big|^{2}\Bigg)^{2}
\end{equation}
In this case the terms $\big|\big\langle \phi_{ik}\big|e^{-it(x_{j}-x_{i})\mathbf{\hat{B}}}\big|\phi_{jk}\big\rangle\big|^{2}$ will not converge to zero, which implies that $\min_{POVM}p_{E}(t)$ in the statement of the theorem does not converge to zero either.
\end{proof}
\begin{Co}{\textbf{QSD with $\sum_{j}\sqrt{\eta_{ij}}<\infty$ for all $j$}:}
\label{eqn:mainresult3}
Theorem \ref{eqn:maintheorem2} may be extended to the cases where the finite mixtures $\boldsymbol{\hat{\rho}}_{i}$ are replaced by infinite mixtures $\boldsymbol{\hat{\rho}}_{i}:=\sum_{j=1}^{\infty}\eta_{ij}\big|\phi_{ij}\big\rangle\big\langle\phi_{ij}\big|$, where now $\sum_{j=1}^{\infty}\eta_{ij}=1$ for all $i$, provided $\sum_{j}\sqrt{\eta_{ij}}<\infty$ for all $i$. The corollary follows by repeating the proof of Theorem \ref{eqn:maintheorem2}, applying the dominated convergence theorem where necessary. 
\end{Co}

%Corollary (\ref{eqn:mainresult3}) gives us a way to work with the spectral decomposition of the operators in the mixtures  $\sum_{i}p_{i}e^{-itx_{i}\mathbf{\hat{B}}}\boldsymbol{\hat{\rho}}_{i}e^{itx_{i}\mathbf{\hat{B}}}$, so long as the sequence $\sqrt{\lambda_{ij}}$ of square-rooted eigenvalues of each $\boldsymbol{\hat{\rho}}_{i}$ is summable with respect to $j$. 

\section{Uncountable Mixtures.  }
\;\;\; Consider the case where instead of a countable mixture, as seen in (\ref{eqn:countablemixture}), we have an uncountable one. 
\begin{equation}
\label{eqn:locenv}
\boldsymbol{\hat{\rho}}_{t}: = \int p(x) \boldsymbol{\hat{\rho}}_{x,t}dx
\end{equation}
where $\mathscr{H}$ is a Hilbert space, $\mathbf{\hat{B}}$ is a self-adjoint operator acting in $\mathscr{H}$, $\boldsymbol{\hat{\rho}}_{x,t} := e^{-it x\mathbf{\hat{B}}}\big|\psi\big\rangle\big\langle \psi\big|e^{it x\mathbf{\hat{B}}}$, $\big|\psi\big\rangle\big\langle \psi\big|$ is an initial state in $\mathcal{S}\big(\mathscr{H}\big)$,  and $\int p(x) dx = 1$. 
In the QSD literature \cite{bae} \cite{qiu} \cite{Montanaro} \cite{knill}, one almost always encounters ensembles of the form $\sum_{i} p_{i} \boldsymbol{\hat{\rho}}_{i}$ where $p_{i}$  is a discrete probability distribution, and the task is to find a POVM that minimizes $\sum_{i}p_{i}Tr\big\{\boldsymbol{\hat{\rho}}_{i} - \boldsymbol{\hat{M}}_{i}\boldsymbol{\hat{\rho}}_{i}\boldsymbol{\hat{M}}^{\dagger}_{i}\big\}$ while satisfying $\sum_{i} \mathbf{\hat{M}}^{\dagger}_{i} \mathbf{\hat{M}}_{i} = \mathbb{I}$.
%For the case of uncountable mixtures (\ref{eqn:locenv}), to discriminate between the $\boldsymbol{\hat{\rho}}_{x,t}$ with high precision, one would expect that $F\big(\boldsymbol{\hat{\rho}}_{x,t},\boldsymbol{\hat{\rho}}_{y,t})$ should go to zero as $t\rightarrow \infty$ for all $x\neq y$. To see that the latter is not the case in general, recall that  $F\big(\boldsymbol{\hat{\rho}}_{x,t},\boldsymbol{\hat{\rho}}_{y,t})= |\big\langle\psi\big|e^{-it(y-x)\mathbf{\hat{B}}}\big|\psi\big\rangle|$. Indeed, for fixed $x\neq y$, $|\big\langle\psi\big|e^{-it(y-x)\mathbf{\hat{B}}}\big|\psi\big\rangle|\rightarrow 0$ as $t\rightarrow \infty$ whenever $\big|\psi\big\rangle \in \mathscr{H}_{rc}$ (Rajchman subspace associated with $\mathbf{\hat{B}}$, see Section 4). However, if for every $t$ we choose $y$ and $x$ so that $x-y= \frac{\alpha}{t}$, then $F\big(\boldsymbol{\hat{\rho}}_{x,t},\boldsymbol{\hat{\rho}}_{y,t}) = |\big\langle\psi\big|e^{-i\alpha\mathbf{\hat{B}}}\big|\psi\big\rangle|$. If $\alpha$ is small, then $F\big(\boldsymbol{\hat{\rho}}_{x,t},\boldsymbol{\hat{\rho}}_{y,t})$ may be close to one. We therefore abandon the idea of discriminating all of the $\boldsymbol{\hat{\rho}}_{x,t}$ from one another and will instead rely on the already existing theory of %
We will rely on the existing theory of QSD for countable mixtures. We will do this by defining the concept of an $N$-mixture associated with the uncountable mixture (\ref{eqn:locenv}).

%To motivate the latter, let us first consider a partition of the support of $p(x)$ into $N$ intervals:  $\cup_{i=1}^{N}\Omega_{i}=\mathbf{supp}\big\{p(x)\big\}$. Using this partition we may rewrite (\ref{eqn:locenv}) as follows. 
%\begin{equation}
%\label{eqn:extentble}
%\int p(x) \boldsymbol{\hat{\rho}}_{x,t}dx = \sum_{i=1}^{N}\int_{\Omega_{i}} p(x) \boldsymbol{\hat{\rho}}_{x,t}dx. 
%\end{equation}
%Next, we define a discrete probability distribution $p_{i}:= \int_{\Omega_{i}}p(x)dx$ and $\bar{p}(x) := \frac{p(x)}{p_{i}}$ and rewrite (\ref{eqn:extentble})
%\begin{equation}
%\sum_{i=1}^{N}\int_{\Omega_{i}} p(x) \boldsymbol{\hat{\rho}}_{x,t}dx= \sum_{i=1}^{N}p_{i}\Lambda_{i,t}\big(\big|\psi\big\rangle\big\langle \psi\big|\big)   
%\end{equation}
%Here %$\Lambda_{i,t}\big(\big|\psi\big\rangle\big\langle \psi\big|\big) :=\int_{\Omega_{i}} \bar{p}(x) e^{-it x\mathbf{\hat{B}}}\big|\psi\big\rangle\big\langle \psi\big|e^{it x\mathbf{\hat{B}}}dx$, so that  $\Lambda_{i,t}\big( \big|\psi\big\rangle\big\langle \psi\big|\big)$ is a density operator (not pure in general). The fact that the $t=0$ state was assumed to be pure was not essential, so in the following definition we replace it by a general state. Let us now formally define an $N$-mixture corresponding to a given uncountable mixture. 
\begin{definition}{$\mathbf{N}$-mixture:}
Let $\boldsymbol{\hat{\rho}}_{t}:=\int p(x) \boldsymbol{\hat{\rho}}_{x,t}dx$ be an uncountable mixture. Given a partition $\cup_{i=1}^{N}\Omega_{i}$ of the support of $p(x)$, let $p_{i}:= \int_{\Omega_{i}}p(x)dx$, $\bar{p}(x) := \frac{p(x)}{p_{i}(t)}$ and $\boldsymbol{\hat{\rho}}_{i,t}:= \int_{\Omega_{i}} \bar{p}(x) \boldsymbol{\hat{\rho}}_{x,t}dx$.  We then call the representation
\begin{equation}
\label{eqn:count2}
 \boldsymbol{\hat{\rho}}_{t} = \sum_{i=1}^{N}p_{i}\boldsymbol{\hat{\rho}}_{i,t}
\end{equation}
an N-mixture associated with the partition.  We emphasize that this is not an approximation but merely a way of rewriting $\boldsymbol{\hat{\rho}}_{t}$; note also that the $\boldsymbol{\hat{\rho}}_{i,t}$ are density operators. 
\end{definition}

Given the representation (\ref{eqn:count2}), we can use the theory of countable mixture QSD to construct a POVM that approximately minimizes
$\sum_{i}p_{i}Tr\big\{\boldsymbol{\hat{\rho}}_{i,t} - \boldsymbol{\hat{M}}_{i,t}\boldsymbol{\hat{\rho}}_{i,t}\boldsymbol{\hat{M}}^{\dagger}_{i,t}\big\}$.  Typically, finding the minimizing POVM exactly is not possible.  However, we estimate the minimal error using the Knill-Barnum bound (\ref{eqn:knill}) \cite{knill} and prove that it goes to zero as $t \rightarrow \infty$, i.e. the QSD problem associated with the UQSD problem is \emph{fully solvable}. The above discussion leads to the following formal definition.

\begin{definition}{QSD for uncountable mixtures:}    \label{eqn:uncountablestatediscrimination}
Let $\mathscr{H}$ be an arbitrary Hilbert space. Consider the uncountably mixed state $\boldsymbol{\hat{\rho}}_{t}:=\int p(x)e^{-it x\mathbf{\hat{B}}}\boldsymbol{\hat{\rho}}e^{it x\mathbf{\hat{B}}}dx$, where $p(x)$ is a probability density, $\boldsymbol{\hat{\rho}} \in \mathcal{S}\big(\mathscr{H}\big)$ is an initial state, and $\mathbf{\hat{B}}$ is a self-adjoint operator acting in $\mathscr{H}$. Furthermore, consider an $N$-mixture of $\boldsymbol{\hat{\rho}}_{t}$ with respect to some partition of the support of $p(x)$, $\cup_{i=1}^{N}\Omega_{i}$. The problem of finding the minimum
\begin{equation}
\label{eqn:uncountdisprob}
\min_{POVM}\sum_{i=1}^{N}p_{i}\Big(1-Tr\big\{\mathbf{\hat{M}}_{i}\boldsymbol{\hat{\rho}}_{i,t}\mathbf{\hat{M}}^{\dagger}_{i}\big\} \Big)
\end{equation}
where $p_{i}: = \int_{\Omega_{i}}p(x)dx$, $\bar{p}_{i}(x):= \frac{p(x)}{p_{i}}$ and $\boldsymbol{\hat{\rho}}_{i,t}:=\int_{\Omega_{i}} \bar{p}_{i}(x)e^{-it x\mathbf{\hat{B}}}\boldsymbol{\hat{\rho}}e^{it x\mathbf{\hat{B}}}dx$ is called the QSD problem associated with the given UQSD problem.
\end{definition}

Proposition 1 cannot be applied to prove complete solvability of the QSD problem arising from uncountable mixtures as described above.  The reason is that $\boldsymbol{\hat{\rho}}_{i,t}$ are not pure states (each one of them has been defined as a mixture of uncountably many pure states), so the techniques used in the proof of Theorem 8 are insufficient in this case. The main hurdle is the fidelity $F\big(\boldsymbol{\hat{\rho}}, \boldsymbol{\hat{\sigma}}\big)$  which is much more difficult to estimate for mixed states. Below we conjecture and motivate necessary and sufficient conditions for full solvability of the QSD problem for uncountable mixtures.  
\begin{definition7}{\textbf{Necessary and Sufficient Conditions for Full Solvability of UQSD for URM}:} 
\label{eqn:resulttwo}
Consider the setup of Definition \ref{eqn:uncountablestatediscrimination}. We conjecture that the UQSD optimization problem induced by a partition $\cup_{i=1}^{N}\Omega_{i}$ is fully solvable as $t \rightarrow \infty$ iff $\boldsymbol{\hat{\rho}}\in\mathcal{S}\big(\mathscr{H}_{rc}\big)$, where $\mathscr{H}_{rc}$ is the Rajchman subspace of the operator $\mathbf{\hat{B}}$.
\end{definition7}
To motivate Conjecture \ref{eqn:resulttwo}, we will need the concept super fidelity, as appearing in the following theorem.
\begin{definition4}{\textbf{Super Fidelity} \cite{superfidelity}:}
\label{eqn:superfidelityq}
For any two density operators $\boldsymbol{\hat{\rho}}$ and $\boldsymbol{\hat{\sigma}}$ 
\begin{equation}
\label{eqn:boiga3000}
 F\big(\boldsymbol{\hat{\rho}}, \boldsymbol{\hat{\sigma}} \big)\leq Tr\big\{\boldsymbol{\hat{\rho}}\boldsymbol{\hat{\sigma}}\big\}+\sqrt{\big(1-Tr\big\{\boldsymbol{\hat{\rho}}^{2}\big\}\big)\big(1-Tr\big\{\boldsymbol{\hat{\sigma}}^{2}\big\}\big)}
 \end{equation}
 The right-hand side of (\ref{eqn:boiga3000}) is called super fidelity. 
\end{definition4}

Let us now consider the uncountable unitarily related mixture $\int p(x)e^{-itx\mathbf{\hat{B}}}\big|\psi\big\rangle\big\langle \psi\big|e^{itx\mathbf{\hat{B}}}dx$ where $\big|\psi\big\rangle \in \mathscr{H}_{rc}$ of $\mathbf{\hat{B}}$. Let $p(x)$ be a symmetric convex combination $p(x) =\frac{1}{2}(p_{1}(x)+p_{2}(x))$ of two probability densities with non-overlapping compact supports. Let $\Delta_{1}\subset\mathbb{R}$ and $\Delta_{2}\subset \mathbb{R}$ denote these supports. With these assumptions, we have the following result. %Now, for $\Delta_{1}$ and $\Delta_{2}$ with any magnitude, i.e. $\delta_{1}:=\big|\Delta_{1}\big|$ and $\delta_{2}:=\big|\Delta_{2}\big|$ and for any %
\begin{definition6}
\label{claimtime}
For any $\varepsilon_{1}>0$ we may find a time domain $\mathscr{T}:=[0,T_{\varepsilon}]$ such that
\begin{equation}
 Tr\Bigg\{ \Bigg(\int p_{i}(x)e^{-itx\mathbf{\hat{B}}}\big|\psi\big\rangle\big\langle \psi\big|e^{itx\mathbf{\hat{B}}}dx\Bigg)^{2}\Bigg\}\geq 1-\varepsilon_{1}   
\end{equation}
for all $t\in\mathscr{T}$ and $i = 1, 2$.
Furthermore, with $\mathscr{T}$ fixed, for any $\varepsilon_{2}>0$ we can choose $\mathbf{dist}\big(\Delta_{1},\Delta_{2}\big)$ such that 
\begin{equation}
Tr\Bigg\{\int p_{1}(x)e^{-itx\mathbf{\hat{B}}}\big|\psi\big\rangle\big\langle \psi\big|e^{itx\mathbf{\hat{B}}}dx\int p_{2}(x)e^{-itx\mathbf{\hat{B}}}\big|\psi\big\rangle\big\langle \psi\big|e^{itx\mathbf{\hat{B}}}dx\Bigg\} <\varepsilon_{2}
\end{equation}
\end{definition6}
\begin{proof}
Fix $\varepsilon_{2}>0$ and let $t^{\prime}\in \mathscr{T}$. Now,
\begin{equation}
Tr\Bigg\{\int p_{1}(x)e^{-itx\mathbf{\hat{B}}}\big|\psi\big\rangle\big\langle \psi\big|e^{itx\mathbf{\hat{B}}}dx\int p_{2}(x)e^{-itx\mathbf{\hat{B}}}\big|\psi\big\rangle\big\langle \psi\big|e^{itx\mathbf{\hat{B}}}dx\Bigg\} =
\end{equation}
\begin{equation}
\label{eqn:psl1}
\int\int p_{1}(x)p_{2}(x)\Big|\big\langle \psi\big|e^{-it(y-x)\mathbf{B}}\big|\psi\big\rangle\Big|^{2} dxdy    
\end{equation}
$\big|\psi\big\rangle\in\mathscr{H}_{rc}$ (of $\mathbf{\hat{B}}$), so $\lim_{\alpha \rightarrow \infty}\big\langle \psi\big|e^{-i\alpha\mathbf{\hat{B}}}\big|\psi\big\rangle = 0$. This implies that there exists a $\delta>0$ such that for all $\alpha> \delta$ we have $|\big\langle \psi\big|e^{-i\alpha\mathbf{\hat{B}}}\big|\psi\big\rangle|\leq\sqrt{\varepsilon_{2}}$. Choosing $\mathbf{dist}\big(\Delta_{1},\Delta_{2}\big)t^{\prime}> \delta$ we have 
\begin{equation}
\label{eqn:psl2}
\int\int p_{1}(x)p_{2}(x)\Big|\big\langle \psi\big|e^{-it(y-x)\mathbf{B}}\big|\psi\big\rangle\Big|^{2} dxdy \leq  \int\int p_{1}(x)p_{2}(x)|\sqrt{\varepsilon_{2}}|^{2} dxdy  = \varepsilon_{2}
\end{equation}
for all $t\in [t^{\prime},T]$.
\end{proof}
Theorem \ref{eqn:superfidelityq} and Claim \ref{claimtime} together imply the following result. 
\begin{equation}
F\Bigg(\int p_{1}(x)e^{-itx\mathbf{\hat{B}}}\big|\psi\big\rangle\big\langle \psi\big|e^{itx\mathbf{\hat{B}}}dx, \;\;\int p_{2}(x)e^{-itx\mathbf{\hat{B}}}\big|\psi\big\rangle\big\langle \psi\big|e^{itx\mathbf{\hat{B}}}dx\Bigg)\leq \varepsilon_{2}+\varepsilon_{1}
\end{equation}
showing that we may approximately solve the UQSD problem for the 2-mixture $\sum_{i=1}^{2}p_{i}\boldsymbol{\hat{\rho}}_{i}$ where $p_{1}=p_{2} = \frac{1}{2}$ and of course $\boldsymbol{\hat{\rho}}_{i}:= \int p_{i}(x)e^{-itx\mathbf{\hat{B}}}\big|\psi\big\rangle\big\langle \psi\big|e^{itx\mathbf{\hat{B}}}dx$. We can proceed similarly for a more general convex linear combinations of probability densities, placing their supports far apart from one another, enough to use the super fidelity bound with $t$ large.  It is clear that $\big|\psi\big\rangle \in \mathscr{H}_{rc}$ of $\mathbf{\hat{B}}$ plays a key role in going from (\ref{eqn:psl1}) to (\ref{eqn:psl2}); without this assumption our conclusion is not attainable.  
\section{Acknowledgements}
Both authors were partially supported by NSF grant DMS 1911358. J. W. was also partially supported by the Simons Foundation Fellowship 823539.

\newpage 

\appendix 

\addcontentsline{toc}{chapter}{APPENDICES}
\section*{Appendix A : Proof of the Knill-Barnum for Infinite-Dimensional Density Operators}\addcontentsline{toc}{chapter}{Appendix A: Notation}
\label{eqn:genknill}
\;\;\; The inequality (\ref{eqn:knill}) was proved in \cite{knill}, under the assumption that the underlying Hilbert space ia finite-dimensional. The proof provided in \cite{knill} involves the inverse of the operator $\sum_{i}p_{i}\boldsymbol{\hat{\rho}}_{i}$ in a crucial manner; this approach is not extendable to the case where the Hilbert space in question is infinite-dimensional since the inverse of a compact operator is in this case unbounded. It is therefore important to prove that the Knill and Barnum bound (\ref{eqn:knill}) may be extended to the infinite-dimensional case. Consider the mixture $\sum_{i}p_{i}\boldsymbol{\hat{\rho}}_{i}$ of density operators, all acting in an infinite dimensional Hilbert space (with $\sum_{i}p_{i}=1$). 
\begin{definition2}
\label{lemma:lemma3}
Let $\boldsymbol{\hat{\rho}}_{i,d} = \sum_{k=1}^{d}\lambda_{ki}\big|\psi_{ki}\big\rangle\big\langle\psi_{ki}\big|$ be a rank d approximation of the operator $\boldsymbol{\hat{\rho}}_{i}$ . Then
\begin{equation}
 \lim_{d\rightarrow \infty}\Big\|\sqrt{\boldsymbol{\hat{\rho}}_{i,d}}\sqrt{\boldsymbol{\hat{\rho}}_{j,d}}\Big\|_{1} = \Big\|\sqrt{\boldsymbol{\hat{\rho}}_{i}}\sqrt{\boldsymbol{\hat{\rho}}_{j}}\Big\|_{1}
\end{equation}
\end{definition2}
\begin{proof}
\begin{equation}
\lim_{d\rightarrow\infty}\bigg|\Big\|\sqrt{\boldsymbol{\hat{\rho}}_{i,d}}\sqrt{\boldsymbol{\hat{\rho}}_{j,d}}\Big\|_{1}- \Big\|\sqrt{\boldsymbol{\hat{\rho}}_{i}}\sqrt{\boldsymbol{\hat{\rho}}_{j}}\Big\|_{1}\bigg| \leq
\end{equation}
\begin{equation}
\lim_{d\rightarrow\infty}\bigg\|\sqrt{\boldsymbol{\hat{\rho}}_{i,d}}\sqrt{\boldsymbol{\hat{\rho}}_{j,d}}- \sqrt{\boldsymbol{\hat{\rho}}_{i}}\sqrt{\boldsymbol{\hat{\rho}}_{j}}\bigg\|_{1} \leq
\end{equation}
\begin{equation}
\lim_{d\rightarrow\infty}\bigg\|\sqrt{\boldsymbol{\hat{\rho}}_{i,d}}\sqrt{\boldsymbol{\hat{\rho}}_{j,d}}- \sqrt{\boldsymbol{\hat{\rho}}_{i}}\sqrt{\boldsymbol{\hat{\rho}}_{j,d}}\bigg\|_{1} +\lim_{d\rightarrow\infty}\bigg\|\sqrt{\boldsymbol{\hat{\rho}}_{i}}\sqrt{\boldsymbol{\hat{\rho}}_{j,d}}- \sqrt{\boldsymbol{\hat{\rho}}_{i}}\sqrt{\boldsymbol{\hat{\rho}}_{j}}\bigg\|_{1} \leq
\end{equation}
\begin{equation}
\lim_{d\rightarrow\infty}\bigg\|\sqrt{\boldsymbol{\hat{\rho}}_{i,d}}- \sqrt{\boldsymbol{\hat{\rho}}_{i}}\bigg\|_{2}\Big\|\sqrt{\boldsymbol{\hat{\rho}}_{j,d}}\Big\|_{2} +\lim_{d\rightarrow\infty}\bigg\|\sqrt{\boldsymbol{\hat{\rho}}_{j,d}}- \sqrt{\boldsymbol{\hat{\rho}}_{j}}\bigg\|_{2}\Big\|\sqrt{\boldsymbol{\hat{\rho}}_{i}}\Big\|_{2} \leq 
\end{equation}
\begin{equation}
\lim_{d\rightarrow\infty}\bigg\|\sqrt{\boldsymbol{\hat{\rho}}_{i,d}}- \sqrt{\boldsymbol{\hat{\rho}}_{i}}\bigg\|_{2} +\lim_{d\rightarrow\infty}\bigg\|\sqrt{\boldsymbol{\hat{\rho}}_{j,d}}- \sqrt{\boldsymbol{\hat{\rho}}_{j}}\bigg\|_{2} = 
\end{equation}
\begin{equation} 
\lim_{d\rightarrow\infty}\bigg\|\sum_{k=d+1}^{\infty}\sqrt{\lambda_{ki}}\big|\psi_{ki}\big\rangle\big\langle \psi_{ki}\big|\bigg\|_{2} +\lim_{d\rightarrow\infty}\bigg\|\sum_{k=d+1}^{\infty}\sqrt{\lambda_{kj}}\big|\psi_{kj}\big\rangle\big\langle \psi_{kj}\big|\bigg\|_{2} =
\end{equation}
\begin{equation}
\lim_{d\rightarrow \infty}\Bigg(\sqrt{\sum_{k=d+1}^{\infty} \lambda_{ki}}+\sqrt{\sum_{k=d+1}^{\infty}\lambda_{kj}}\Bigg) =\sqrt{\lim_{d\rightarrow \infty}\sum_{k=d+1}^{\infty} \lambda_{ki}}+\sqrt{\lim_{d\rightarrow \infty}\sum_{k=d+1}^{\infty} \lambda_{kj}} =
\end{equation}
\begin{equation}
\sqrt{0}+\sqrt{0} = 0   
\end{equation}
\end{proof}
We are now ready to generalize Theorem \ref{eqn:knill1}.

\begin{definition4}{\textbf{Knill-Barnum bound extended to infinite dimensional density operators}:}
\label{eqn: knillbarnumgeneralized}
Let $\sum_{i=1}^{N} p_{i} \boldsymbol{\hat{\rho}}_{i}$ be a finite mixture of infinite-dimensional density operators. Then, the Knill-Barnum bound (\ref{eqn:knill}) applies.
\end{definition4}
\begin{proof}
 Starting from the optimization problem $ \min_{POVM}\sum_{i=1}^{N}\sum_{j; j\neq i}^{N}p_{i}Tr\{\mathbf{\hat{M}}_{j}\boldsymbol{\hat{\rho}}_{i}\mathbf{\hat{M}}^{\dagger}_{j}\} $, notice that we may rewrite each $\boldsymbol{\hat{\rho}}_{i}$ as the limit of a sequence of finite rank operators. To see this, first, we diagonalize the $\boldsymbol{\hat{\rho}}_{i}$, i.e. we write $\boldsymbol{\hat{\rho}}_{i} = \sum_{k=1}^{\infty}\lambda_{ki}\big|\psi_{ki}\big\rangle\big\langle\psi_{ki}\big|$. Here
 , the $\lambda_{ki}$ are the eigenvalues of $\boldsymbol{\hat{\rho}}_{i}$. A $d$-rank approximation of $\boldsymbol{\hat{\rho}}_{i}$ is therefore $\boldsymbol{\hat{\rho}}_{i,d}:= \sum_{k=1}^{d}\lambda_{ki}\big|\psi_{ki}\big\rangle\big\langle\psi_{ki}\big|$  and indeed 
 \begin{equation} \lim_{d\rightarrow \infty} \Big\|\boldsymbol{\hat{\rho}}_{i,d} -  \boldsymbol{\hat{\rho}}_{i}\Big\|_{1}= 
 \lim_{d\rightarrow \infty}\Big\|\sum_{k = d+1}^{\infty}\lambda_{ki}\big|\psi_{ki}\big\rangle\big\langle \psi_{ki}\big|\Big\|_{1} \leq \lim_{d\rightarrow \infty} \sum_{k = d+1}^{\infty}|\lambda_{ki}| = 0.
 \end{equation}
To proceed we must first demonstrate
\begin{equation}
\min_{POVM}\sum_{i=1}^{N}\sum_{j; j\neq i}^{N}p_{i}Tr\{\mathbf{\hat{M}}_{j}\boldsymbol{\hat{\rho}}_{i}\mathbf{\hat{M}}_{j}^{\dagger}\} = \lim_{d\rightarrow\infty}\min_{POVM}\sum_{i=1}^{N}\sum_{j; j\neq i}^{N}p_{i}Tr\{\mathbf{\hat{M}}_{j}\boldsymbol{\hat{\rho}}_{i,d}\mathbf{\hat{M}}_{j}^{\dagger}\}.
\end{equation}
To prove the latter it is enough to show that
\begin{equation} 
\lim_{d\rightarrow\infty}\bigg|\min_{POVM}\sum_{i=1}^{N}\sum_{j; i\neq j}^{N}p_{i}Tr\{\mathbf{\hat{M}}_{j}\boldsymbol{\hat{\rho}}_{i}\mathbf{\hat{M}}_{j}^{\dagger}\} - \min_{POVM}\sum_{i=1}^{N}\sum_{j; i\neq j}^{N}p_{i}Tr\{\mathbf{\hat{M}}_{j}\boldsymbol{\hat{\rho}}_{i,d}\mathbf{\hat{M}}_{j}^{\dagger}\}\bigg|=  0  
\end{equation}
We have
\begin{equation}
\bigg|\min_{POVM}\sum_{i=1}^{N}\sum_{j; j\neq i}^{N}p_{i}Tr\{\mathbf{\hat{M}}_{j}\boldsymbol{\hat{\rho}}_{i}\mathbf{\hat{M}}_{j}^{\dagger}\} - \min_{POVM}\sum_{i=1}^{N}\sum_{j; i\neq j}^{N}p_{i}Tr\{\mathbf{\hat{M}}_{j}\boldsymbol{\hat{\rho}}_{i,d}\mathbf{\hat{M}}_{j}^{\dagger}\}\bigg| \leq 
\end{equation}

\begin{equation}
\max_{POVM}\bigg|\sum_{i=1}^{N}\sum_{j; j\neq i}^{N}p_{i}Tr\{ \mathbf{\hat{M}}_{j}\big(\boldsymbol{\hat{\rho}}_{i} -\boldsymbol{\hat{\rho}}_{i,d}\big)\mathbf{\hat{M}}_{j}^{\dagger}\} \bigg| = \max_{POVM}\bigg|\sum_{i=1}^{N}\sum_{j; j\neq i}^{N}p_{i}Tr\{\mathbf{\hat{M}}_{j}^{\dagger} \mathbf{\hat{M}}_{j}\big(\boldsymbol{\hat{\rho}}_{i} -\boldsymbol{\hat{\rho}}_{i,d}\big)\} \bigg|\leq
\end{equation}

\begin{equation}
\max_{POVM}\sum_{i=1}^{N}\sum_{j; i\neq j}^{N}p_{i}\Big\|\mathbf{\hat{M}}_{j}^{\dagger}\mathbf{\hat{M}}_{j}\big(\boldsymbol{\hat{\rho}}_{i} -\boldsymbol{\hat{\rho}}_{i,d}\big) \Big\|_{1} \leq\max_{POVM}\sum_{i=1}^{N}\sum_{j; j\neq i}^{N}p_{i}\big\|\mathbf{\hat{M}}_{j}^{\dagger} \mathbf{\hat{M}}_{j}\big\|\Big\|\big(\boldsymbol{\hat{\rho}}_{i} -\boldsymbol{\hat{\rho}}_{i,d}\big)  \Big\|_{1}\leq 
\end{equation}
\begin{equation}
\max_{POVM}\sum_{i=1}^{N}\sum_{j; j\neq i}^{N}p_{i}\Big\|\big(\boldsymbol{\hat{\rho}}_{i} -\boldsymbol{\hat{\rho}}_{i,d}\big)  \Big\|_{1}  = \sum_{i=1}^{N}\sum_{j; j\neq i}^{N}p_{i}\Big\|\big(\boldsymbol{\hat{\rho}}_{i} -\boldsymbol{\hat{\rho}}_{i,d}\big) \Big\|_{1} = \end{equation}
\begin{equation}
\sum_{i=1}^{N}\sum_{j; j\neq i}^{N}p_{i}\sum_{k = d+1}^{\infty}\lambda_{ki} \leq N\sum_{i=1}^{N}p_{i}\sum_{k = d+1}^{\infty}\lambda_{ki}
\end{equation}
Now, since
\begin{equation}
\lim_{d\rightarrow \infty} N\sum_{i=1}^{N}p_{i}\sum_{k = d+1}^{\infty}\lambda_{ki} = 
\end{equation}
\begin{equation}
N\sum_{i=1}^{N}p_{i}\lim_{d\rightarrow \infty}\sum_{k = d+1}^{\infty}\lambda_{ki}   =  N\sum_{j=1}^{N}0 = 0
\end{equation}
it follows that 
\begin{equation}
\lim_{d\rightarrow \infty}\bigg|\min_{POVM}\sum_{i=1}^{N}\sum_{j; j\neq i}^{N}p_{i}Tr\{\mathbf{\hat{M}}_{j}\boldsymbol{\hat{\rho}}_{i}\mathbf{\hat{M}}_{j}^{\dagger}\} - \min_{POVM}\sum_{i,j; i\neq j}^{N}p_{i}Tr\{\mathbf{\hat{M}}_{j}\boldsymbol{\hat{\rho}}_{i,d}\mathbf{\hat{M}}_{j}^{\dagger}\}\bigg| =0 
\end{equation}
which in turn implies 
\begin{equation} \min_{POVM}\sum_{i=1}^{N}\sum_{j; j\neq i}^{N}p_{j}Tr\{\mathbf{\hat{M}}_{j}\boldsymbol{\hat{\rho}}_{i}\mathbf{\hat{M}}_{j}^{\dagger}\} = \lim_{d\rightarrow \infty}\min_{POVM}\sum_{i=1}^{N}\sum_{j; j\neq i}^{N}p_{i}Tr\{\mathbf{\hat{M}}_{j}\boldsymbol{\hat{\rho}}_{i,d}\mathbf{\hat{M}}_{j}^{\dagger}\}.
\end{equation}
Let us now introduce a normalizing constant $\alpha_{i,d}: = Tr\big\{ \boldsymbol{\hat{\rho}}_{i,d}  \big\}$. Using this normalization and the Knill-Barnum bound (\ref{eqn:knill}) \cite{knill}, we have
\begin{equation}
\lim_{d\rightarrow \infty}\min_{POVM}\sum_{i=1}^{N}\sum_{j; j\neq i}^{N}p_{j}\alpha_{i,d}Tr\{\mathbf{\hat{M}}_{j}\alpha_{i,d}^{-1}\boldsymbol{\hat{\rho}}_{i,d}\mathbf{\hat{M}}_{j}^{\dagger}\}\leq
\lim_{d\rightarrow \infty}\max_{k}\big(\alpha_{k,d}\big)\min_{POVM}\sum_{i=1}^{N}\sum_{j; j\neq i}^{N}p_{i}Tr\{\mathbf{\hat{M}}_{j}\alpha_{i,d}^{-1}\boldsymbol{\hat{\rho}}_{i,d}\mathbf{\hat{M}}_{j}^{\dagger}\}\leq\end{equation}
\begin{equation}
\lim_{d\rightarrow \infty}\max_{k}\big(\alpha_{k,d}\big) \sum_{i=1}^{N}\sum_{j; j\neq i}^{N}\sqrt{p_{i}p_{j}}\sqrt{F(\alpha_{i,d}^{-1}\boldsymbol{\hat{\rho}}_{i,d}, \alpha_{j,d}^{-1}\boldsymbol{\hat{\rho}}_{j,d})} = 
\end{equation}
\begin{equation}
\lim_{d\rightarrow \infty}\max_{k}\big(\alpha_{k,d}\big) \sum_{i=1}^{N}\sum_{j; j\neq i}^{N}\sqrt{p_{i}p_{j}}\Big\|\sqrt{\alpha_{i,d}^{-1}\boldsymbol{\hat{\rho}}_{i,d}}\sqrt{ \alpha_{j,d}^{-1}\boldsymbol{\hat{\rho}}_{j,d}}\Big\|_{1}=
\end{equation}
\begin{equation}
\sum_{i=1}^{N}\sum_{j; j\neq i}^{N}\sqrt{p_{i}p_{j}}\lim_{d\rightarrow \infty}\max_{k}\big(\alpha_{k,d}\big)\sqrt{\alpha_{i,d}^{-1}\alpha_{j,d}^{-1}}\Big\|\sqrt{\boldsymbol{\hat{\rho}}_{i,d}}\sqrt{ \boldsymbol{\hat{\rho}}_{j,d}}\Big\|_{1}=
\end{equation}
\begin{equation}
\sum_{i=1}^{N}\sum_{j; j\neq i}^{N}\sqrt{p_{i}p_{j}}\Big(\lim_{d\rightarrow \infty}\max_{k}\big(\alpha_{k,d}\big)\sqrt{\alpha_{i,d}^{-1}\alpha_{j,d}^{-1}}\Big)\Big(
\lim_{d\rightarrow\infty}\Big\|\sqrt{\boldsymbol{\hat{\rho}}_{i,d}}\sqrt{ \boldsymbol{\hat{\rho}}_{j,d}}\Big\|_{1}\Big) \leq
\end{equation}
\begin{equation}
\sum_{i=1}^{N}\sum_{j; j\neq i}^{N}\sqrt{p_{i}p_{j}}\Big(\lim_{d\rightarrow \infty}\sqrt{\alpha_{i,d}^{-1}\alpha_{j,d}^{-1}}\Big)\Big(
\lim_{d\rightarrow\infty}\Big\|\sqrt{\boldsymbol{\hat{\rho}}_{i,d}}\sqrt{ \boldsymbol{\hat{\rho}}_{j,d}}\Big\|_{1}\Big) 
\end{equation}
\begin{equation}
\sum_{i=1}^{N}\sum_{j; j\neq i}^{N}\sqrt{p_{i}p_{j}}\sqrt{F(\boldsymbol{\hat{\rho}}_{i}, \boldsymbol{\hat{\rho}}_{j})}
\end{equation}
where we have used lemma \ref{lemma:lemma3} and the fact that $\lim_{d\rightarrow} \alpha_{k,d} = 1$ for all $k$ in the final equality. 
\end{proof}
\section*{Appendix B : Quantum Chernoff Bounds \cite{szkola}}\addcontentsline{toc}{chapter}{Appendix A: Notation}
\label{eqn:cherr}
One of the main results in \cite{szkola} is the following relationship, useful for the study of Asymptotic QSD. 
\begin{equation}
\frac{1}{3}\xi_{QCB}\Big(\big\{\boldsymbol{\hat{\rho}}_{i}\big\}_{i=1}^{N}\Big)\leq -\lim_{n\rightarrow \infty} \frac{\log\Big(\min_{POVMP}p_{E}(n)\Big)}{n} \leq \xi_{QCB}\Big(\big\{\boldsymbol{\hat{\rho}}_{i}\big\}_{i=1}^{N}\Big)
\end{equation}
where $\xi_{QCB}$ is the quantum Chernoff bound for an $N$ mixture, defined as
\begin{equation}
 \xi_{QCB}\Big(\big\{\boldsymbol{\hat{\rho}}_{i}\big\}_{i=1}^{N}\Big):= \min_{i,j}\xi_{QCB}\big(\boldsymbol{\hat{\rho}}_{i},\boldsymbol{\hat{\rho}}_{j} \big)  
\end{equation}
where 
\begin{equation}
\xi_{QCB}\big(\boldsymbol{\hat{\rho}}_{i},\boldsymbol{\hat{\rho}}_{j} \big) = -\log\big(\min_{0\leq s \leq 1} Tr\big\{\boldsymbol{\hat{\rho}}_{i}^{s}\boldsymbol{\hat{\rho}}_{j}^{1-s}   \big\}  \big)
\end{equation}

\end{document}